\documentclass[10pt, conference, letterpaper]{IEEEtran}

\usepackage{cite}
\usepackage{url}
\usepackage{color}
\usepackage{amsthm}

\usepackage{graphicx}
\usepackage{algorithm}
\usepackage[noend]{algorithmic}
\usepackage{subfigure}
\usepackage{amssymb, amsmath,graphicx,charter, latexsym}
\usepackage{enumerate}

\newtheorem{definition}{Definition}
\newtheorem{lemma}{Lemma}

\newtheorem{theorem}{Theorem}

\begin{document}
\title{Broadcasting Real-Time Flows in Integrated Backhaul and Access 5G Networks}
\author{Aria HasanzadeZonuzy$^*$,
I-Hong Hou$^*$ and
Srinivas Shakkottai$^*$\\
$^*$Dept. of Electrical and Computer Engineering, Texas A\&M University\\
{\small Email:\{azonuzy, ihou, sshakkot\}@tamu.edu}}

%\IEEEauthorblockN{Aria HasanzadeZonuzy}
%\IEEEauthorblockA{I-Hong Hou}
% }
\maketitle

\begin{abstract}
This paper studies the problem of broadcasting real-time flows in multi-hop wireless networks.  We consider that each packet has a stringent deadline, and each node in the network obtains some utility based on the number of packets delivered to it on time for each flow.  We propose a distributed protocol, the delegated-set routing (DSR) protocol, that incurs virtually no overhead of coordination among nodes. We also develop distributed algorithms that aim to maximize the total system utility under DSR.  The utility of our DSR protocol and distributed algorithms are demonstrated by both theoretical analysis and simulation results, where we show that our algorithms achieve better performance even when compared against centralized throughput optimal policies.
\end{abstract}

\section{Introduction}
\label{section:introduction}

Mutli-hop broadcasting in wireless networks, which entails disseminating information to every device in the system via retransmissions at multiple nodes, is an important mechanism to coordinate devices in networked systems.  Furthermore, many applications of broadcast communications are safety-critical, and timely deliveries of information is crucial to maintain the robustness and safety of the system.  For example, multi-hop broadcasting is needed to disseminate timely safety information among connected vehicles in vehicular ad hoc networks (VANETs), to announce control decisions in networked control systems and Internet of Things (IoT), and to exchange locations and flight paths among unmanned aerial vehicles (UAVs) for Unmanned Aircraft System Traffic Management (UTM).  %Since these emerging networked systems will be deployed over large areas, multi-hop transmissions are  often needed to disseminate information to devices that are far from the source. %Furthermore, many applications of networked systems are safety-critical, and timely deliveries of information is critical to maintain the robustness and safety of the system.

The cellular infrastructure that will enable these time-critical broadcast wireless applications will be 5G networks that are currently being designed to support ultra-low latency, ultra-high throughput communications.   These networks will utilize the highly directional and high bandwidth mm-wave band, which suffers from high attenuation and sensitively to fading.  This requires the relatively dense deployment of small base stations at spacings of about $250$ m.  However, providing fiber backhaul to all of these base stations is prohibitively expensive.  An important development in this context is \textit{Integrated Access and Backhaul} (IAB) \cite{iab,iab2}, under which there are a few base stations with fiber backhaul that act as gateways to many others that are connected via a mm-wave wireless mesh backhaul.  This mm-wave backhaul creates a directional wireless network between the nodes, but routing across these is highly dynamic and subject to the vagaries of the wireless channel.  The same mm-wave spectrum also is used to provide access to end-users, i.e., both access and backhaul are integrated over mm-wave.

% On the other hand, one of the most important multi-hop wireless networks is \textit{Integrated Access and Backhaul} (IAB) networks \cite{iab}. In these networks, there are mixture of fiber-connected and non-fiber-connected base stations. One of the data going through this network is VR gaming that most of the content of a multi-player games are identical to players. Besides, the players would be spread over the network. Thus, broadcasting the common content of the data through the network would be more efficient.}

Motivated by the above features of emerging networks, this paper studies the problem of designing algorithms for broadcasting real-time flows with strict per-packet end-to-end deadlines in directional wireless mesh networks.  Here, real-time flow imposes a strict deadline for each of its packets, and packets that cannot be delivered before their respective deadlines are dropped from the system.   From the IAB perspective, our goal is to ensure that each broadcast packet is delivered to an appropriate IAB base station before its deadline, at which point it is immediately transmitted to its respective end user.  Each IAB node in the network then obtains some utility based on the time-average number of on-time packets that it receive from each flow. The goal of this paper is to maximize the total timely-utility of the whole network.

% Each node in the network obtains some utility based on the time-average number of packets that it receives on time from each flow. The goal of this paper is to maximize the total utility of the whole network. 

% in the emerging millimeter wave band (mmWave).  The mmWave band has the potential to achieve very high link capacity due to its large bandwidth opportunity, and is expected to be an important band for fifth-generation (5G) wireless networks.  The use of multi-user multiple-input and multiple-output (MU-MIMO) \cite{mimo,mimo2,mimo3} devices along with beamforming implies that transmission and reception from multiple nodes can occur simultaneously using highly directional (but lossy) links.  Each real-time flow imposes a strict deadline for each of its packets, and packets that cannot be delivered before their respective deadlines are dropped from the system. Each node in the network obtains some utility based on the time-average number of packets that it receives on time from each flow. The goal of this paper is to maximize the total utility of the whole network.

There are several important challenges that need to be addressed for broadcasting real-time flows in such multi-hop mmWave networks. First, since it is difficult to coordinate a large network in real-time, centralized algorithms that require the instant knowledge of the state of each node and packet are usually infeasible to implement.  Hence, we need distributed algorithms, where each node makes decisions using its local information.  Second, as mentioned above, transmissions in the mmWave band can be unreliable.  Finally, broadcasting algorithms need to explicitly address the deadline requirement of each flow.

\subsection*{Main Results and Organization}
In this paper, we propose a new protocol for broadcasting in multi-hop mmWave networks, namely, the \emph{delegated-set routing} (DSR) protocol.  DSR has two important features: First, it is a distributed protocol where all the required coordination among nodes can be conveyed in the headers of packets once the topology of the network is known. Hence, there is virtually no overhead of coordination after topology creation process.  Second, DSR allows each node to dynamically change its transmission strategies based on the deadlines of its packets and random events, such as transmission failures, it experiences. 

Relaxing the link utilization constraint (number of tranmissions allowed per time slot) to an average one, and using dual decomposition techniques, we also propose a distributed algorithm that aims to maximize the total system-wide utility under DSR.  This algorithm only requires minimal and infrequent information exchange among nodes.  We analytically prove that our algorithm achieves the optimal total utility under an average link capacity constraint. The key novelty lies in a natural decomposition into packet-by-packet and link-by-link updates that need minimal coordination.  These lead to a steepest-ascent-type control associated with each packet, and a sub-gradient type of update at links.  This algorithm also gives rise to a simple index policy when link utilization constraints of all links need to be satisfied at every instant.

We evaluate our algorithms through simulations on representative network graphs.  We compare our algorithms against recent studies on throughput optimal algorithms, including one that is designed specifically for broadcast, and one that is universal in terms of being able to support unicast, multicast and broadcast.  We show that despite some of these algorithms being centralized and complex, our algorithm, which is designed specifically for simplicity and delay optimality, achieves better performance.

The paper is organized as follows. 
Section \ref{section:related} reviews existing studies on broadcasting and multi-hop networks. Section \ref{section:model} describes our system model for multi-hop networks with real-time broadcast flows. 
%Section \ref{section:challenges} discusses an optimal centralized solution based on Markov Decision Process, and demonstrates such a solution is infeasible to implement in actual networks. 
%Section \ref{section:model} discusses the system model and goals, and the intractability of an Markov Decision Process (MDP) based solution.
Section \ref{section:structure} describes the additional structure imposed by the DSR protocol, as well as an epoch-wise approach to policy selection.
Section \ref{section:overview}  applies dual-decomposition, which turns out to be the basis of our distributed algorithm.  
Section \ref{section:algo} proposes distributed algorithms that optimize DSR, as well as the index policy that can ensure hard capacity constraints are met.  
Section \ref{section:simulation} presents our simulation results.  Finally, Section \ref{section:conclusion} concludes the paper.

\section{Related Work}
\label{section:related}

Broadcasting/multicasting is a fundamental functionality of networks, and has been studied in a substantial body of literature.  One of the earliest policies for broadcasting/multicasting in ad hoc networks is via flooding  \cite{Ho99, Johnson01}.  However, such policies can lead to severe packet collision frequency, and excessive redundant retransmissions, as shown by Ni et al. \cite{Ni99}.   Gandhi et al. \cite{Gandhi03} and Huang et al. \cite{Huang07} have shown that the problem of minimizing delay in wireless ad hoc networks is NP-hard, and have proposed approximation algorithms aiming to reduce delay. These studies rely on centralized algorithms. 

%Hou et al. \cite{hou2008adapcode} and Yeow et al. \cite{yeow2009minimizing} have proposed using network coding for wireless broadcast, but their solutions lack provable per-packet guarantees. 
%Further, these studies focus on omnidirectional wireless transmissions, and cannot be directly applied to the highly directional mmWave band.
%the proposed several approximation policies to address the  problem with the assistance of pre-constructed broadcasting trees.  {\color{red}Nevertheless, these works could not apply to the more general multi-hop networks.}

There has been much interest in throughput optimal broadcasting/multicasting.   For instance, Sarkar and Tassiulas \cite{Tass02} proposed a scheduling and routing policy that relies on pre-computed spanning trees, which might be difficult to maintain and compute in scalable sized networks.  Ho and Viswanathan \cite{Ho05} and Yuan et al. \cite{Yuan06} propose network coding based policies in the context, which, however, leads to additional computation complexity.   Zhang et al. \cite{Zhang13} and Sinha et al. \cite{Sinha16} consider multi-hop broadcasting problems in Directed Acyclic Graphs (DAG), which are not applicable to networks with arbitrary topology.  Sinha et al. \cite{Sinha16-2} also propose a centralized throughput optimal broadcasting policy for networks with arbitrary topology, which might be difficult to deploy in a large scale system.  Furthermore, the throughput maximization focus of all the above does not directly allow for meeting stringent deadline guarantees.

Given the rising application of wireless networks to safety-critical and realtime applications, there has been much recent interest in deadline constrained multi-hop communication.  Xiong et al. \cite{Xiong11} proposed a delay-aware throughput optimal policy for multi-hop networks.  Their policy, however, can not provide stringent delay guarantees.  Mao et al. \cite{Mao13} propose a hard deadline guaranteed policy, under the assumption that all routes in the network are fixed.  Li and Eryilmaz \cite{Li12} consider serving flows with stringent deadlines in a multi-hop system, and their proposed framework can be extended to incorporate routing decisions.  However, their policies are heuristic, and optimality cannot be shown.   Singh and Kumar \cite{Kumar15} relax the deadline constrained optimization problem in the manner of the  Whittle's relaxation for multi-armed bandits, and proposed decentralized optimal solutions.   However, both it and the above body of work on deadline constrained communication only considers unicast traffic, and it is not clear how it applies to broadcasting/multicasting networks.

\section{System Model}
\label{section:model}

We consider a multi-hop network that consists of $N$ wireless nodes operating in the mmWave band motivated by the IAB system.   Here, the nodes correspond to fixed IAB base stations, and the network topology is known to all nodes.   The available spectrum is divided into multiple half-duplex channels, and nodes can use these channels to send and receive packets from multiple nodes simultaneously.  Furthermore, these channels are \emph{directional} in that transmissions on different links do not interfere with each other, consistent with empirical observations in IAB test deployments~\cite{iab2}.  These links can have different constraints on the supportable number of transmissions in each time slot, as well as their reliabilities.    

Time is slotted and numbered as $t=1,2,\dots.$   We assume that link $l$ can transmit $T_l$ packets in each time slot, and that each transmission will be successfully received by the receiver with probability $P_l$.   At the end of each time slot, the receiver sends an aggregated ACK indicating which packets it has successfully received in the time slot to the transmitter.   Where we need to indicate the transmitter and the receiver of a link, we use $l=n\rightarrow m$ to indicate that link $l$ has transmitter $n$ and receiver $m$.

We consider $F$ real-time broadcast flows, using $s_f$ to indicate the source node of flow $f$.  At the beginning of each time slot $t,$ $a_f(t)$ packets of flow $f$ arrive at node $s_f$.  We assume that $[a_f(1), a_f(2),\dots]$ is a sequence of i.i.d. random variables with mean $A_f$.  Moreover, each flow $f$ specifies a per-packet end-to-end deadline of $D_f$ time slots.  Packets from flow $f$ are only useful for $D_f$ time slots from their respective arrival times at their source nodes, and are dropped from the network when they expire.  Due to communication constraints, it is likely that some nodes cannot receive all packets from each flow.  We therefore measure the performance of node $n$ on flow $f$ by its \emph{timely-throughput}, defined as the long-term average number of packets from flow $f$ that are successfully delivered to node $n$ within the deadline. % end-to-end deadline. 

%  Let $x_{n,f}(t)$ be the number of packets from $f$ that are delivered to $n$ at time $t$.{\color{blue} Hence, $x_{n,f}(t)$ are is the number of packets which survived the deadline constraint.} The timely-throughput of node $n$ on flow $f$, denoted by $q_{n,f}$, can be formally written as
% \begin{equation}
%     q_{n,f}:=\liminf_{T \to \infty}\frac{\sum_{t=1}^T x_{n,f}(t)}{T}.
% \end{equation} 

Let $\Omega$ be a set of stationary packet scheduling policies.  Hence, given the state of the system consisting of the locations and expiry times of all existing packets, a policy $\omega\in \Omega$ is a rule that decides which packet to transmit on what link, subject to communication constraints.  For each stationary policy $\omega\in \Omega$, let $x^\omega_{n,f}(t)$ be the number of packets from $f$ that are delivered to $n$ at time $t$ under $\omega,$ i.e., these are the packets that survived the deadline constraint.  Also, let $\epsilon^\omega_{l,f}(t)$ be the number of packets from flow $f$ transmitted over link $l$ at time $t$ under $\omega$.  Since $\omega$ is a stationary policy, and all packets that expire are immediately dropped, we can define
$$\mu^\omega_{n,f}:=\liminf_{T \to \infty}\frac{\sum_{t=1}^T x^\omega_{n,f}(t)}{T}$$
as the timely-throughput of node $n$ on flow $f$ under $\omega$, and 
$$\bar{\epsilon}^\omega_{l,f}:=\limsup_{T \to \infty}\frac{\sum_{t=1}^T \epsilon^\omega_{l,f}(t)}{T}$$
 as the average number of transmissions for flow $f$ over link $l$ under $\omega$. 

Now, finding the optimal total utility with respect to timely-throughputs over all the $N$ nodes under DSR is equivalent to finding the stationary policy that maximizes the total timely-utility under link utilization constraints, which can be written as
\begin{align}
    \mbox{\textbf{Relaxed\ }} &\mbox{\textbf{Timely-Utility Maximization (R-TUM)}}\nonumber\\ 
    \mbox{Max }&\sum_{n=1}^N\sum_{f=1}^FU_{n,f}(\mu^{\omega}_{n,f})\label{eq:relaxed1}\\
    \mbox{s.t. }&\omega \in \Omega,\label{eq:relaxed2}\\
     & \sum_{f=1}^F\bar{\epsilon}_{l,f}^{\omega} \leq T_l, \forall l. \label{eq:relaxed3}
\end{align}

Notice that whereas the R-TUM problem above requires each delivered packet to satisfy its deadline constraint, it only requires that the \emph{long-term average} number of transmissions over link $l$, $\sum_{f=1}^F\bar{\epsilon}_{l,f}^{\omega}$ be no larger than $T_l.$  This link utilization constraint relaxation is in the same manner as \cite{Kumar15}.  In a practical system, such a relaxation might be akin to imposing an average transmit power constraint rather than a hard one.    We will first design policies that pertain to this relaxed link-utilization constraint.  Using the insights gained, we will also develop a policy that enforces a \emph{hard link-utilization constraint,} i.e., 
\begin{equation}
\sum_{f=1}^F\epsilon^\omega_{l,f}(t)\leq T_l, \forall l, t. \label{eq:util3}
\end{equation}

Solving the R-TUM problem could be posed as a Markov Decision Process (MDP), where the state of the system at any given point of time consists of the locations and expiry times of all existing packets.  However, such a solution is infeasible to implement in practice.  First, it is straightforward to show that the number of different system states is at least doubly exponential in $N$, and hence standard algorithms for finding the optimal MDP-based solution will result in prohibitive complexity.   Second, even after one finds the optimal MDP-based solution, it may  be impossible to implement it in a distributed fashion, since the complete state needs to be known at each node.  In what follows, we impose additional structure on the policy space to  render it tractable.

\section{A Structured Approach to Real-Time Broadcasting}
\label{section:structure}

We now introduce two elements of structure to the policy space to enable its solution as a distributed convex optimization problem.

\subsection{Delegated-Set Routing (DSR)}

Ensuring a per-packet deadline guarantee requires that we retain flexibility in routing to dynamically choose the next hop node for a packet based on current state.  Thus, source routing on a per-packet basis is not satisfactory.  However, for distributed implementation, we also need to ensure that there is no ambiguity as to which neighboring node is responsible for transmitting a packet to a given node.  We resolve these seemingly opposite requirements via a protocol that we term \emph{delegated-set routing (DSR)}.   

For each node $n$ that possesses a packet $i$ at time $t,$ we define the delegated-set of node $n$ as the subset of nodes that $n$ is responsible for forwarding packets, possibly through multi-hop transmissions.  First, to ensure routing flexibility,  whenever a node $n$ decides to forward a packet to another node $m$, node $n$ delegates a subset of its own delegated-set to $m,$ and specifies this subset in the packet header.  If the transmission is successful, this subset is removed from the delegated-set of $n,$ since it is now the responsibility of $m$ to forward the packet to this subset.  Second, in order to avoid duplicate transmissions (ambiguity on which node should transmit a given packet), the DSR protocol requires that the delegated-sets of different nodes for the same packet are chosen to be disjoint. 

\begin{figure}[htb]
\vspace{-0.1in}
\center
\includegraphics[width=1.6in]{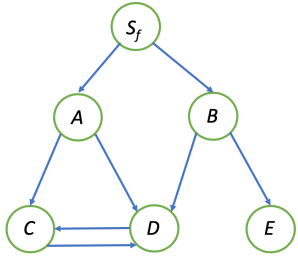}
\caption{An example illustrating DSR.}\label{fig:ex_dynamic_routing}
\vspace{-0.1in}
\end{figure}
To illustrate how DSR works, consider the network as shown in Fig. \ref{fig:ex_dynamic_routing}.  When a packet arrives at the source node $s_f,$ the delegated set of $s_f$ is every node in the network, since it is the responsibility of $s_f$ to broadcast the packet to the entire network.  Suppose in the first time slot, $s_f$ transmits the packet to $A,$ and delegates the subset $\{A,C,D\}$ to $A$.  If the transmission is successful, the delegated-set of $s_f$ becomes $\{s_f, B, E\}$, while the delegated-set of $A$ is $\{A,C,D\}$.  In the next time slot, when $s_f$ transmits the packet to $B$, it needs to delegate the subset $\{B, E\}$ to $B$.  In particular, $s_f$ cannot include node $D$ in the delegated-set for $B$, since $D$ is already in the delegated-set of $A$.

% There are two attractive features of DSR.  First, it allows distributed implementation.  In the above example, when node $A$ receives the packet and the delegated-set of $\{A, C, D\}$, it knows for sure that node $D$ will never obtain the packet from $B$.  Hence, there is no need for explicit coordination between $A$ and $B$. 

% Second, DSR allows intermediate nodes to dynamically adjust routing decisions based on random events, such as transmission failure, observed in the network.  A distributed algorithm for finding optimal decisions will be proposed in the next two sections.

We note that the ability to dynamically adjust routing decisions is an important feature that distinguishes DSR from many existing studies on multi-hop broadcasting, such as \cite{Sinha16-2} and \cite{Tass02}.  These studies adopt source-routing, where the source node determines the routing decision of each packet, and intermediate nodes cannot change the decision.  As $s_f$ cannot foresee whether the transmissions from $A$ to $C$ will be successful, it cannot take an  optimal routing decision.

\subsection{Epoch-wise Stationary Policies}
Our second aspect of adding structure to the policy space is to expand it from $\Omega$, the set of all stationary policies, to the set of all \emph{epoch-wise stationary policies}.  In an epoch-wise stationary policy, time is divided into epochs of equal length. The epoch-wise stationary policy adopts a stationary policy $\omega^+[i]$ in each epoch $i$. The duration of an epoch is chosen to be large enough so that the average performance of $\omega^+[i]$ in epoch $i$ is not influenced by the system state at the beginning of the epoch. Specifically, an epoch-wise stationary policy is defined as follows:

\begin{definition}
An epoch-wise stationary policy is a sequence of stationary policies $\omega^+= (\omega[i])_{i=1}^{\infty}, \omega[i] \in \Omega\}$, where $\omega[i]$ is used in epoch $i$. The length of an epoch is chosen so that, under 
$\omega^+$, $$\mu^{\omega^+}_{n,f}:=\liminf_{T \to \infty}\frac{\sum_{t=1}^T x^{\omega^+}_{n,f}(t)}{T}=\liminf_{I\to \infty}\frac{\sum_{i}^I \mu^{\omega[i]}_{n,f}}{I},$$ and $$\bar{\epsilon}^{\omega^+}_{l,f}:=\limsup_{T \to \infty}\frac{\sum_{t=1}^T \epsilon^{\omega^+}_{l,f}(t)}{T}=\limsup_{I\to \infty}\frac{\sum_{i}^I \bar{\epsilon}^{\omega[i]}_{l,f}}{I}.$$
\end{definition}

We can now define $\Omega^+$ as the set of all epoch-wise stationary policies. For each epoch-wise stationary policy $\omega^+$, let $\gamma^{\omega^+}:=[[\mu^{\omega^+}_{n,f}, 1\leq n\leq N, 1\leq f\leq F], [\bar{\epsilon}^{\omega^+}_{l,f}, 1\leq l\leq L, 1\leq f\leq F]]$ be the vector of timely-throughputs and average link uses under $\omega^+$. Also, let $\Gamma:=\{\gamma^{\omega^+}|\omega^+\in \Omega^+\}$ be the set of attainable vectors of timely-throughputs and average link uses under all epoch-wise stationary policies. An important advantage of considering the policy space $\Omega^+$ is that $\Gamma$ is a convex set.

\begin{lemma}
$\Gamma$ is convex.
\label{lem:gamma}
\end{lemma}
\begin{proof}
The proof is is provided in section \ref{sec:apen}.
\end{proof}

Since $\Gamma$ is a convex set, it is straightforward to verify that the optimization problem (\ref{eq:relaxed1})--(\ref{eq:relaxed3}) subject the policy space $\Gamma$ is a convex optimization problem.

\section{Solution Overview}
\label{section:overview}

Although the problem \textbf{R-TUM}, (\ref{eq:relaxed1}) -- (\ref{eq:relaxed3}), is convex, solving it directly remains challenging because there is no simple characterization of $\Gamma$. In this section, we present a general framework of solving \textbf{R-TUM} through dual decomposition. The exact distributed algorithm will be presented in the next section.

\subsection{Dual Problem Formulation}

Let $\lambda_l$ be the Lagrange multiplier with respect to the constraint $\sum_{f=1}^F\bar{\epsilon}_{l,f}^{\omega^+} \leq T_l$ in (\ref{eq:relaxed3}), and $\lambda$ be the vector of all $\lambda_l$, $l=1,2,\dots, L$.  The Lagrangian of \textbf{R-TUM} is then 
\begin{align}
    \mathcal{L}(\gamma^{\omega^+}, \lambda) = &\sum_{n=1}^N\sum_{f=1}^FU_{n,f}(\mu^{\omega^+}_{n,f})-\sum_{l=1}^L\lambda_l\left(\sum_{f=1}^F\bar{\epsilon}_{l,f}^{\omega^+}- T_l\right), \label{eq:lag}
\end{align}
and the dual objective function is 
\begin{equation}
    \mathcal{D}(\lambda)=\max_{\gamma \in \Gamma}\mathcal{L}(\gamma, \lambda). \label{eq:dual}
\end{equation}

The dual problem of \textbf{R-TUM} is to find a non-negative vector $\lambda$ that minimizes $\mathcal{D}(\lambda)$.

We first show that strong duality holds for \textbf{R-TUM}.

\begin{theorem}
Let $\mathcal{P}^*$ be the optimal solution to \textbf{R-TUM}, and $\mathcal{D}^*:=\min_{\lambda: \lambda_l\geq 0, \forall l}\mathcal{D}(\lambda)$, then $\mathcal{P}^*=\mathcal{D}^*$.
\label{th:rel_utility}
\end{theorem}

\begin{proof}
The proof is provided in section \ref{sec:apen}.
\end{proof}

Hence, solving \textbf{R-TUM} is equivalent to solving the dual problem, which consists of two steps: First, given a vector $\lambda$, we need to find the dual objective function $\mathcal{D}(\lambda)$. Second, we need to find the vector $\lambda$ that minimizes $\mathcal{D}(\lambda)$.

\subsection{Packet-By-Packet Decomposition for the Dual Objective}
We first present an iterative algorithm that finds $\mathcal{D}(\lambda)=\max_{\gamma \in \Gamma}\mathcal{L}(\gamma, \lambda)$ for a given $\lambda$ using the steepest ascent algorithm.   For each stationary policy $\omega$, let $\gamma^\omega$ be  defined to be the vector of timely-throughputs and link usages under $\omega.$   Then the steepest ascent algorithm constructs a sequence of epoch-wise stationary policies that ultimately converges to the optimal epoch-wise stationary policy.  The algorithm proceeds as follows:
%Starting with an arbitrary stationary policy $\omega_1$. Our steepest ascent algorithm constructs a sequence of stationary policies $\omega_2^+, \omega_3^+, \dots$ and a sequence of epoch-wise stationary policies $\omega_1^+, \omega_3^+, \dots$ as follows:
\begin{enumerate}
\item Set $k\leftarrow 1$
\item Let $\omega_k^+$ be the round-robin epoch-wise stationary policy that  follows the sequence $\{\omega_1, \omega_2,\dots,\omega_k, \omega_1, \omega_2,\dots,\omega_k,\dots\}$. 
\item Let $\omega_{k+1}$ be the stationary policy that maximizes the directional derivative, $\nabla \mathcal{L}(\gamma^{\omega_k^+},\lambda)\cdot\gamma^{\omega_{k+1}}.$%\footnote{For each stationary policy $\omega$, $\gamma^\omega$ is defined to be the vector of timely-throughputs and link usages under $\omega$.}.
\item Set $k\leftarrow k+1$ and repeat step 2.
\end{enumerate}

Based on our construction of $\omega_k^+$, we have $\gamma^{\omega_k^+}=\frac{\sum_{j=1}^k\gamma^{\omega_j}}{k}$. Therefore $\gamma^{\omega_{k+1}^+}-\gamma^{\omega_{k}^+}=\frac{\gamma^{\omega_{k+1}}-\gamma^{\omega_k^+}}{k+1}$. Effectively, for each $k$, our steepest ascent algorithm finds $\omega_{k+1}^+$ that maximizes the directional derivative $\nabla \mathcal{L}(\gamma^{\omega_k^+},\lambda)\cdot(\gamma^{\omega_{k+1}^+}-\gamma^{\omega_{k}^+})$ among all epoch-wise stationary policies with step size $\frac{1}{k+1}$.  Following the analysis presented in Boyd~et~al. \cite{boyd2004convex} Section 9.4.3, it is straightforward to show the following:

\begin{theorem}
Under our steepest ascent algorithm, $\mathcal{L}(\gamma^{\omega_k^+}, \lambda)$ converges to $\mathcal{D}(\lambda)$, as $k\rightarrow \infty$.
\end{theorem}

Notice that the critical step in our steepest ascent policy is to find $\omega_{k+1}$ that maximizes $\nabla L(\gamma^{\omega_k^+},\lambda)\cdot\gamma^{\omega_{k+1}}$.  We have
\begin{align*}
&\nabla \mathcal{L}(\gamma^{\omega_k^+},\lambda)\cdot\gamma^{\omega_{k+1}}\\
=&\sum_{n,f}\frac{\partial}{\partial \mu_{n,f}}\mathcal{L}(\gamma^{\omega_k^+},\lambda)\mu_{n,f}^{\omega_{k+1}}
+\sum_{l,f}\frac{\partial}{\partial \bar{\epsilon}_{l,f}}\mathcal{L}(\gamma^{\omega_k^+},\lambda)\bar{\epsilon}_{l,f}^{\omega_{k+1}}\\
=&\sum_{f=1}^F\left\{\sum_{n=1}^NU'_{n,f}(\mu_{n,f}^{\omega_k^+})\mu_{n,f}^{\omega_{k+1}}-\sum_{l=1}^L\lambda_l\bar{\epsilon}_{l,f}^{\omega_{k+1}}\right\}.
\end{align*}

This naturally gives us a flow-by-flow decomposition in the sense that $\nabla \mathcal{L}(\gamma^{\omega_k^+},\lambda)\cdot\gamma^{\omega_{k+1}}$ can be maximized by maximizing
\begin{equation}
\sum_{n=1}^NU'_{n,f}(\mu_{n,f}^{\omega_k^+})\mu_{n,f}^{\omega_{k+1}}-\sum_{l=1}^L\lambda_l\bar{\epsilon}_{l,f}^{\omega_{k+1}}\label{eq:flow-by-flow}
\end{equation}  
for each flow $f$ individually. Moreover, note that, after normalizing with the average packet arrival rate of flow $f$, $\mu_{n,f}^{\omega_{k+1}}$ is the average delivery per-packet from flow $f$ to node $n$, and $\bar{\epsilon}_{l,f}^{\omega_{k+1}}$ is the average number of transmissions per packet over link $l$ for flow $f$. 

For each packet $i$ from flow $f,$ let $y_{n,f,i}$ be a random variable representing the event that packet $i$ is successfully delivered to node $n$  within its deadline of $d_f$.   Also, let $z_{l,f,i}$ be the random variable indicating the number times that link $l$ transmits $i.$  Then $\mathop{\mathbb{E}}[y_{n,f,i}]$ is the success probability that packet $i$ is delivered to node $n,$ while $\mathop{\mathbb{E}}[z_{l,f,i}]$ is the expected number of times that link $l$ transmits $i$.   Therefore, from (\ref{eq:flow-by-flow}), maximizing $\nabla \mathcal{L}(\gamma^{\omega_k^+},\lambda)\cdot\gamma^{\omega_{k+1}}$ can be achieved by maximizing 
\begin{equation}
\sum_{n=1}^NU'_{n,f}(\mu_{n,f}^{\omega_k^+})\mathop{\mathbb{E}}[y_{n,f,i}] - \sum_{l=1}^L\lambda_l\mathop{\mathbb{E}}[z_{l,f,i}]\label{eq:packet-by-packet}
\end{equation} 
for each packet $i$. 

We note that such packet-by-packet decomposition allows distributed algorithms for finding the optimal solution since, instead of considering the system state as a whole, each packet only needs to maximize (\ref{eq:packet-by-packet}) on its own, without considering the states of other packets.

%its own $\sum_{n=1}^NU'_{n,f}(q_{n,f}^{\omega_k^+})\mathop{\mathbb{E}}[y^f_{i,n}] - \sum_{l=1}^L\lambda_l\mathop{\mathbb{E}}[z^f_{i,l}]$ without considering the states of other packets.

\subsection{Link-by-Link Update for the Dual Problem}
After finding $\mathcal{D}(\lambda)$, we now proceed to find the solution to the dual problem, $\min_{\lambda:\lambda_l\geq 0, \forall l} \mathcal{D}(\lambda)$. Our solution is based on the subgradient method. We first find the subgradient of $\mathcal{D}(\lambda)$.

\begin{theorem}
Let $\gamma(\lambda)=[[\mu_{n,f}(\lambda)], [\bar{\epsilon}_{l,f}(\lambda)]]:=\arg\max_{\gamma\in \Gamma}\mathcal{L}(\gamma, \lambda)$, then the $L$-dimensional vector $[T_l-\sum_{f=1}^F\bar{\epsilon}_{l,f}(\lambda)]$ is a subgradient for $\mathcal{D}(\lambda)$.
\label{th:lambda}
\end{theorem}

\begin{proof}
The proof is provided in section \ref{sec:apen}.
\end{proof}

The subgradient method finds the optimal $\lambda$ that minimizes $\mathcal{D}(\lambda)$ iteratively. Starting with an arbitrary vector $\lambda(1)$, the subgradient method finds $\lambda(k+1)=[\lambda_l(k+1)]$ by setting
\begin{equation}
\lambda_l(k+1) = {\left[\lambda_l(k) - \beta_k \left(T_l-\sum_{f=1}^F\bar{\epsilon}_{l,f}(\lambda(k))\right)\right]}^+,
\label{eq:update}
\end{equation}
where $x^+:=\max\{0, x\}$.

\begin{theorem}
If the sequence $\beta_k$ is chosen so that $\beta_k\geq 0, \forall k$, $\sum_{k=1}^\infty \beta_k=\infty$, and $\lim_{k\to\infty}\beta_k=0$, then $\mathcal{D}(\lambda(k))\to \min_{\lambda:\lambda_l\geq 0,\forall l}\mathcal{D}(\lambda)$, as $k\to\infty$.
\end{theorem}
\begin{proof}
This is the direct result of Theorem 8.9.2 in \cite{theorem}.
\end{proof}

Recall that $\sum_{f=1}^F\bar{\epsilon}_{l,f}(\lambda(k)))$ is the average number of transmissions that link $l$ makes. Therefore, for link $l$ to update  $\lambda_l$ by (\ref{eq:update}), link $l$ only needs to know its own link constraint and the number of transmissions it makes. Hence, this subgradient method allows for a distributed update of $\lambda_l$.
\iffalse
This decomposition can be interpreted as follows: For each packet from flow $f$, the system obtains a payoff of $U'_{n,f}(q_{n,f})$ if the packet is delivered to node $n$ before its deadline. On the other side, for each packet transmitted over link $l$, the system needs to pay a price of $\lambda_l$. We then aim to maximize the expected total net payoff over all packets, which is expected payoff minus expected price. With this decomposition, each packet only needs maximize its own net payoff, without the need to consider the system state.
\fi

\section{Optimization of DSR}
\label{section:algo}

Under DSR, the transmission strategy for a node having a packet $i$ consists of two parts: determining which node to transmit the packet $i$ to, and determining what delegated-set to assign to the receiver. In this section, we discuss the optimal transmission strategy that maximizes (\ref{eq:packet-by-packet}) under the design of DSR.

Fix a packet $i$ from flow $f$. For each subset of nodes $\pi$, let $\mathcal{L}_\pi$ be the set of links whose transmitter and receiver are both in $\pi$. Also, for each node $n$, subset of nodes $\pi$, and integer $\tau\in[0,d_f]$, define 
\begin{align}
W_f (n, \pi, \tau) &= \nonumber\\
&\max \left(\sum_{k\in \pi}U'_{k,f}(\mu_{k,f}^{\omega_k^+})\mathop{\mathbb{E}}[y_{k,f,i}] - \sum_{l\in \mathcal{L}_\pi}\lambda_l\mathop{\mathbb{E}}[z_{l,f,i}]\right)
\end{align} 
if node $n$ receives the packet $i$ and delegated-set $\pi$, and the packet $i$ has $\tau$ time slots before meeting its deadline.

By the definition of $W_f(n, \pi, \tau)$, finding the optimal transmission strategy that maximizes (\ref{eq:packet-by-packet})  
%$\sum_{n=1}^Nr_{i,n}\mathop{\mathbb{E}}[y_{i,n}] - \sum_{l=1}^L\lambda_l\mathop{\mathbb{E}}[z_{i,l}]$ 
is equivalent to finding the value of $W_f(s_f, \{1,2,\dots,N\}, d_f)$, as well as the transmission strategy that achieves it.

We use dynamic programming to find $W_f(n, \pi, \tau)$. Suppose node $n$ receives the packet $i$ and delegated-set $\pi$, and packet $i$ has $\tau$ time slots before meeting its deadline. Also suppose that node $n$ decides to transmit the packet to $m$ and designates the delegated-set $\pi^m$ to $m$. If the transmission is successful, then, in the next time slot, node $n$ has a delegated-set of $\pi-\pi^m$, node $m$ has a delegated-set of $\pi^m$, and packet $i$ has $\tau-1$ time slots before its deadline. By the definition of $W_f(\cdot)$, we have, given that the transmission is successful,
\begin{align}
    &\max \left(\sum_{k\in \pi}U'_{k,f}(\mu_{k,f}^{\omega_k^+})\mathop{\mathbb{E}}[y_{k,f,i}] - \sum_{l\in \mathcal{L}_{\pi}}\lambda_l\mathop{\mathbb{E}}[z_{l,f,i}]\right) \nonumber\\
    &= W_f(n, \pi-\pi^m, \tau-1)+W_f(m, \pi^m, \tau-1)-\lambda_{n\rightarrow m}.    \label{eq:DP1}
\end{align}

On the other hand, if the transmission fails, then, in the next time slot, node $n$ still has the delegated-set $\pi$ and packet $i$ has $\tau-1$ time slots before its deadline. Given that the transmission fails, we have
\begin{align}
    &\max  \left(\sum_{k\in \pi}U'_{k,f}(\mu_{k,f}^{\omega_k^+})\mathop{\mathbb{E}}[y_{k,f,i}] - \sum_{l\in \mathcal{L}_{\pi}}\lambda_l\mathop{\mathbb{E}}[z_{l,f,i}]\right)  \nonumber\\
    &= W_f(n, \pi, \tau-1)-\lambda_{n\rightarrow m}. \label{eq:DP2}
\end{align}

Since each transmission from $n$ to $m$ succeeds with probability $P_{n\rightarrow m}$, we have, given that $n$ transmits packet $i$ and assigns delegated-set $\pi^m$ to $m$, 
\begin{align}
    &\max  \left(\sum_{k\in \pi}U'_{k,f}(\mu_{k,f}^{\omega_k^+})\mathop{\mathbb{E}}[y_{k,f,i}] - \sum_{l\in \mathcal{L}_{\pi}}\lambda_l\mathop{\mathbb{E}}[z_{l,f,i}]\right) \nonumber\\
    &= P_{n\rightarrow m}\times\text{(\ref{eq:DP1})}+(1-P_{n\rightarrow m})\times\text{(\ref{eq:DP2})}.    \label{eq:DPfull}
\end{align}

% $$\max (\sum_{k\in \pi}r_{i,k}\mathop{\mathbb{E}}[y_{i,k}] - \sum_{l\in \mathcal{L}_{\pi}}\lambda_l\mathop{\mathbb{E}}[z_{i,l}])
% =P_{n\rightarrow m}\times$(\ref{eq:DP1})$+(1-P_{n\rightarrow m})\times$(\ref{eq:DP2}).

Based on the above analysis, we can write down the following iterative equation:
\begin{align}
    &W_f(n, \pi, \tau) =\max\{W_f(n, \pi, \tau-1), \nonumber\\
    &\max_{m, \pi^m:m\in \pi^m, \pi^m\subset\pi} [P_{n\rightarrow m}(W_f(n, \pi-\pi^m, \tau-1)\nonumber\\
    &+W_f(m, \pi^m, \tau-1))+(1-P_{n\rightarrow m})W_f(n, \pi, \tau-1)\nonumber\\
    &-\lambda_{n\rightarrow m}]\}, \label{eq:DP_eq}
\end{align}
with boundary condition
\begin{align}
    W_f(n, \pi, 0) = r_{i,n}=U'_{n,f}(\mu_{n,f}), \label{eq:DP_boundary}
\end{align}
where the term $W_f(n, \pi, \tau-1)$ in (\ref{eq:DP_eq}) represents the case when $n$ does not transmit the packet at all.  Eq. (\ref{eq:DP_eq}) and (\ref{eq:DP_boundary}) allows a dynamic programming algorithm to find $W_f(n, \pi, \tau)$ for all $f, n, \pi,$ and $\tau$.   
As we will show in Section~\ref{section:simulation}, our algorithm can be easily carried out in medium-sized networks.

\subsection{Index-DSR for Per-Time-Slot Link Constraint}
\label{section:index}

The Dynamic Program in (\ref{eq:DP_eq}) can be directly combined with the dual decomposition in Section~\ref{section:overview} to achieve the optimal solution of \textbf{R-TUM} problem under DSR.  In this section, we further propose an index policy that satisfies the per-time-slot link utilization constraint $\sum_{i,v} \epsilon_{i,v,l}(t)\leq T_l$, for all $t$, of the original \textbf{TUM} problem. The index-DSR policy would be to transmit the maximum number of packets among all possible packets to be transmitted so that the per-time-slot link constraint is not violated.

We make several changes to the dynamic program and the dual decomposition technique.  First, we change the iterative equation (\ref{eq:DP_eq}) to
\begin{align}
    &W_f(n, \pi, \tau) =\nonumber\\
    &\max_{m, \pi^m:m\in \pi^m, \pi^m\subset\pi} [P_{n\rightarrow m}(W_f(n, \pi-\pi^m, \tau-1)\nonumber\\
    &+W_f(m, \pi^m, \tau-1))+(1-P_{n\rightarrow m})W_f(n, \pi, \tau-1)\nonumber\\
    &-\lambda_{n\rightarrow m}]\}, \label{eq:DP_index_eq}
\end{align}
as long as there is a link from $n$ to another node in $\pi$, and 
\begin{align}
    &W_f(n, \pi, \tau) =W_f(n, \pi, \tau-1), \label{eq:DP_index_eq2}
\end{align}
otherwise. In other words, we force each node $n$ to find a link to transmit each packet.  We also define $m^*(n,\pi, \tau)$ and $\pi^{m*}(n,\pi, \tau)$ as the optimal $m$ and $\pi^m$ that achieves $W_f(n,\pi, \tau)$.  We note that, since we now force each node $n$ to find a link to transmit each packet, it is possible that $W_f(n, \pi, \tau)$ is negative for some $(n, \pi, \tau)$.

Second, in each time slot $t$ and for each link $n\rightarrow m$, we find all packets possessed by $n$ with delegated-set $\pi$, $\tau$ slots until their respective deadlines, and $m^*(n,\pi, \tau) = m$.  We sort these packets in descending order of $W_f(n,\pi, \tau)$, and let $\epsilon_{n\rightarrow m}'(t)$ be the number of these packets with $W_f(n,\pi, \tau)>0$.  In other words, $\epsilon_{n\rightarrow m}'(t)$ is the number of packets whose optimal strategy yields a positive return by transmitting over the link ${n\rightarrow m}$. After sorting these packets, link ${n\rightarrow m}$ simply transmit the first $T_{n\rightarrow m}$ packets.  Finally, the price of each link is updated by (\ref{eq:update}).

\section{Simulation Results}
\label{section:simulation}
\begin{figure*}[htbp]
\centering
\begin{minipage}{.32\textwidth}
\centering
\includegraphics[width=1.1\columnwidth]{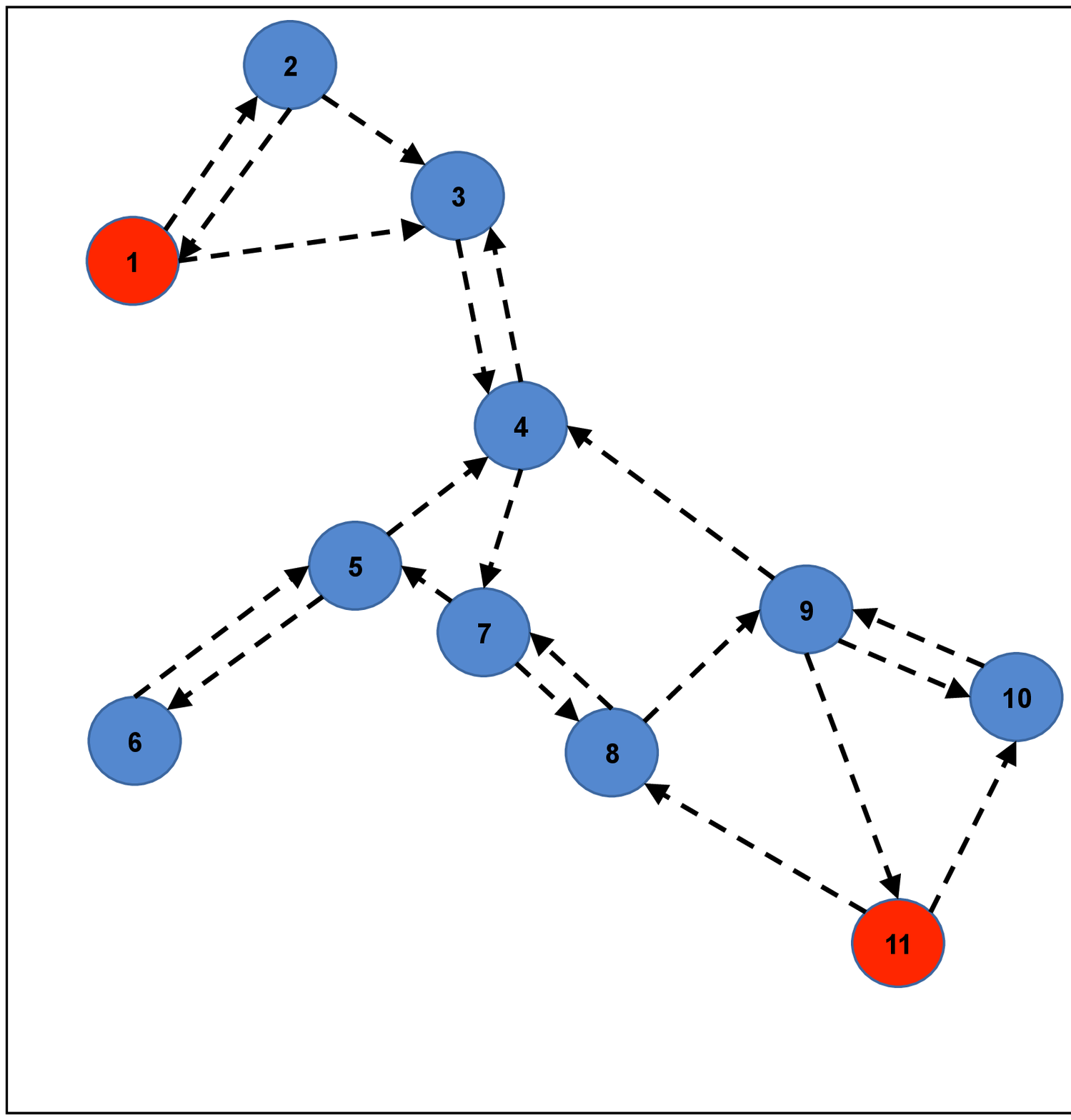}
\caption{Scenario 1: $11-$node network}
\label{fig:sc1}
\end{minipage}\hfill
\begin{minipage}{.32\textwidth}
\centering
\includegraphics[width=1.0\columnwidth]{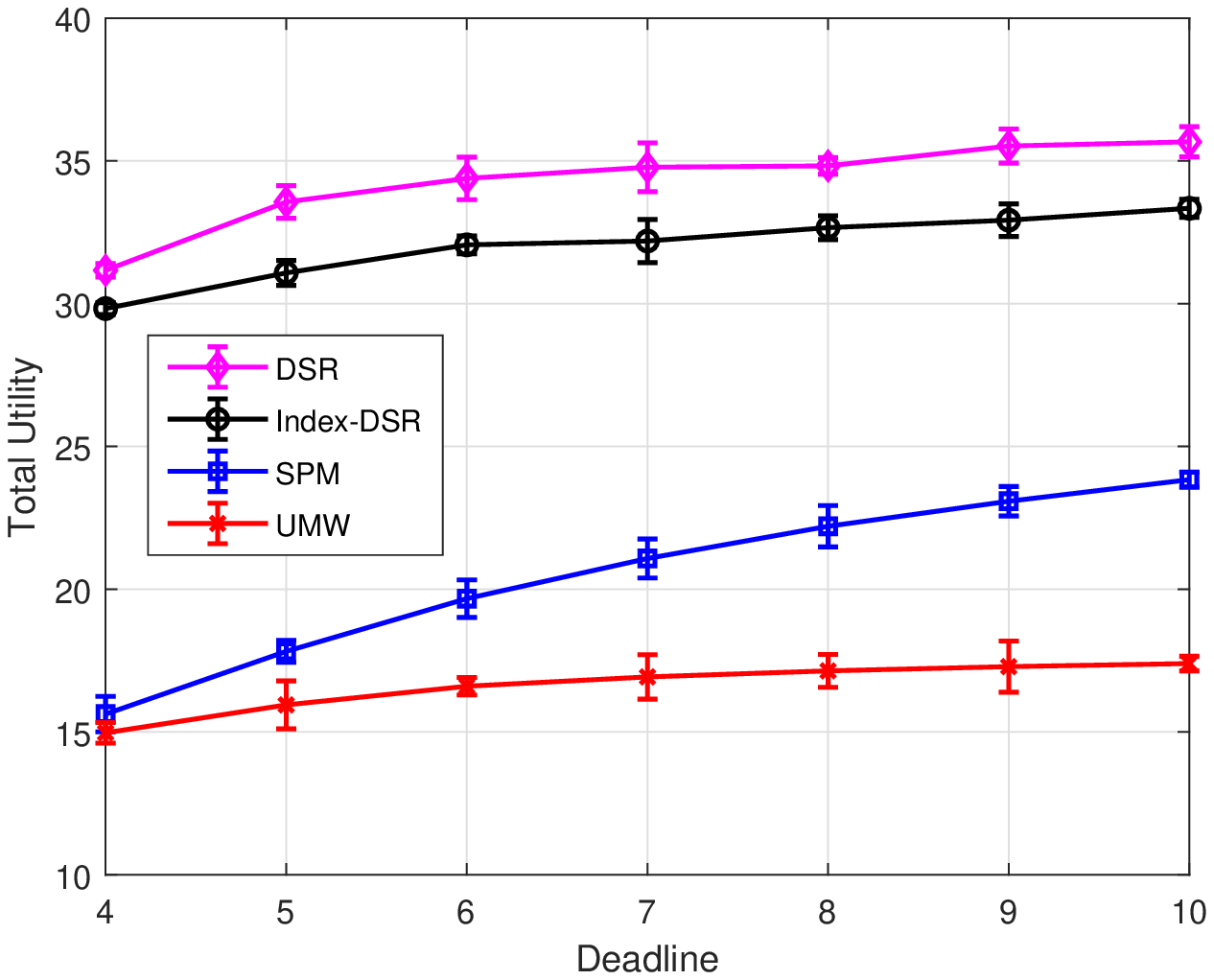}
\caption{Scenario 1: Linear Utility Function}% for scenario in figure (\ref{fig:sc1})}
\label{fig:sc1_lin}
\end{minipage}\hfill
\begin{minipage}{.32\textwidth}
\centering
\includegraphics[width=1.0\columnwidth]{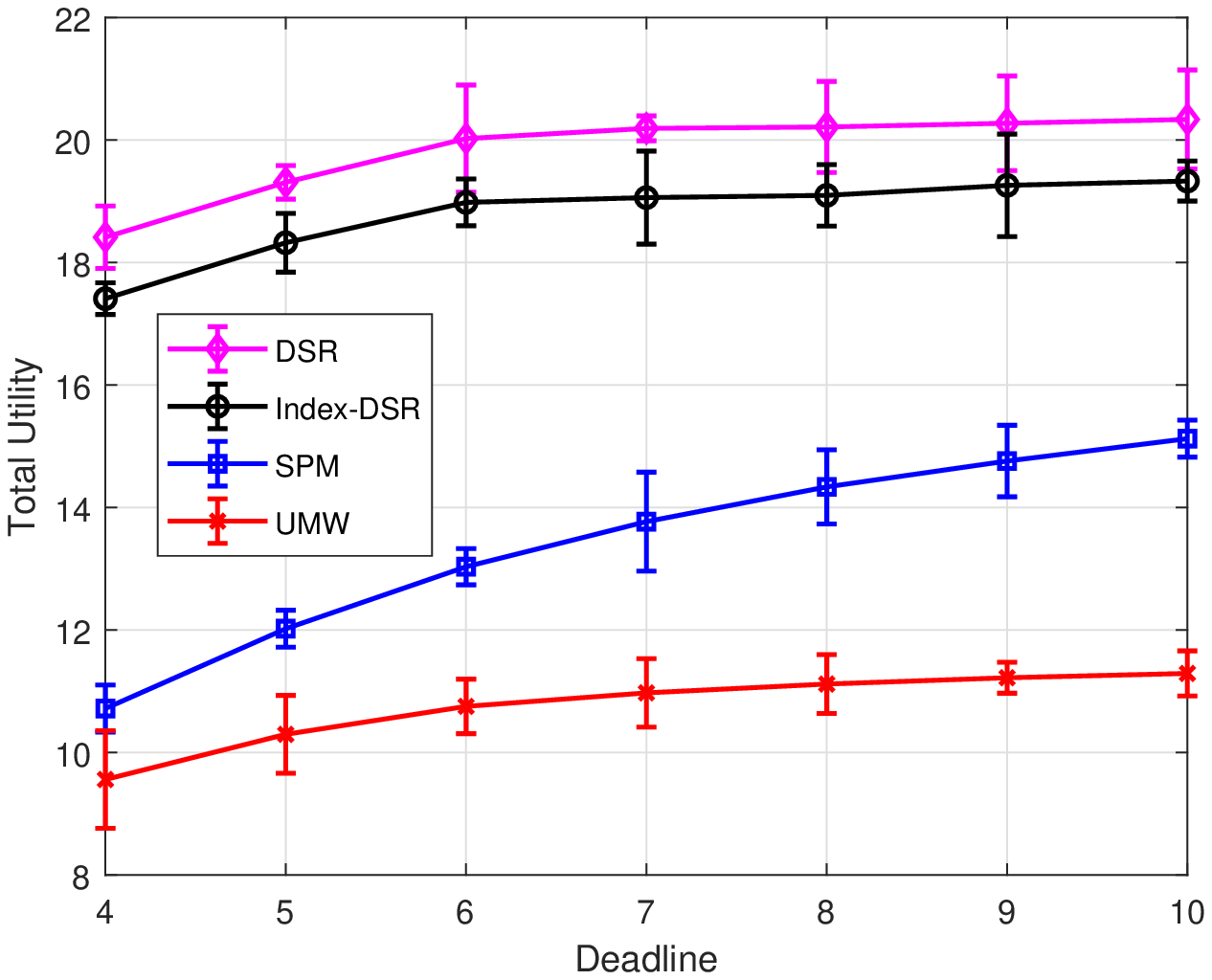}
\caption{Scenario 1: Logarithmic Utility Function}% for scenario in figure (\ref{fig:sc1})}% for odd and even vehicles.}
\label{fig:sc1_log}
\end{minipage}
%\vspace{0.1in}
\begin{minipage}{.32\textwidth}
\centering
\includegraphics[width=1.1\columnwidth]{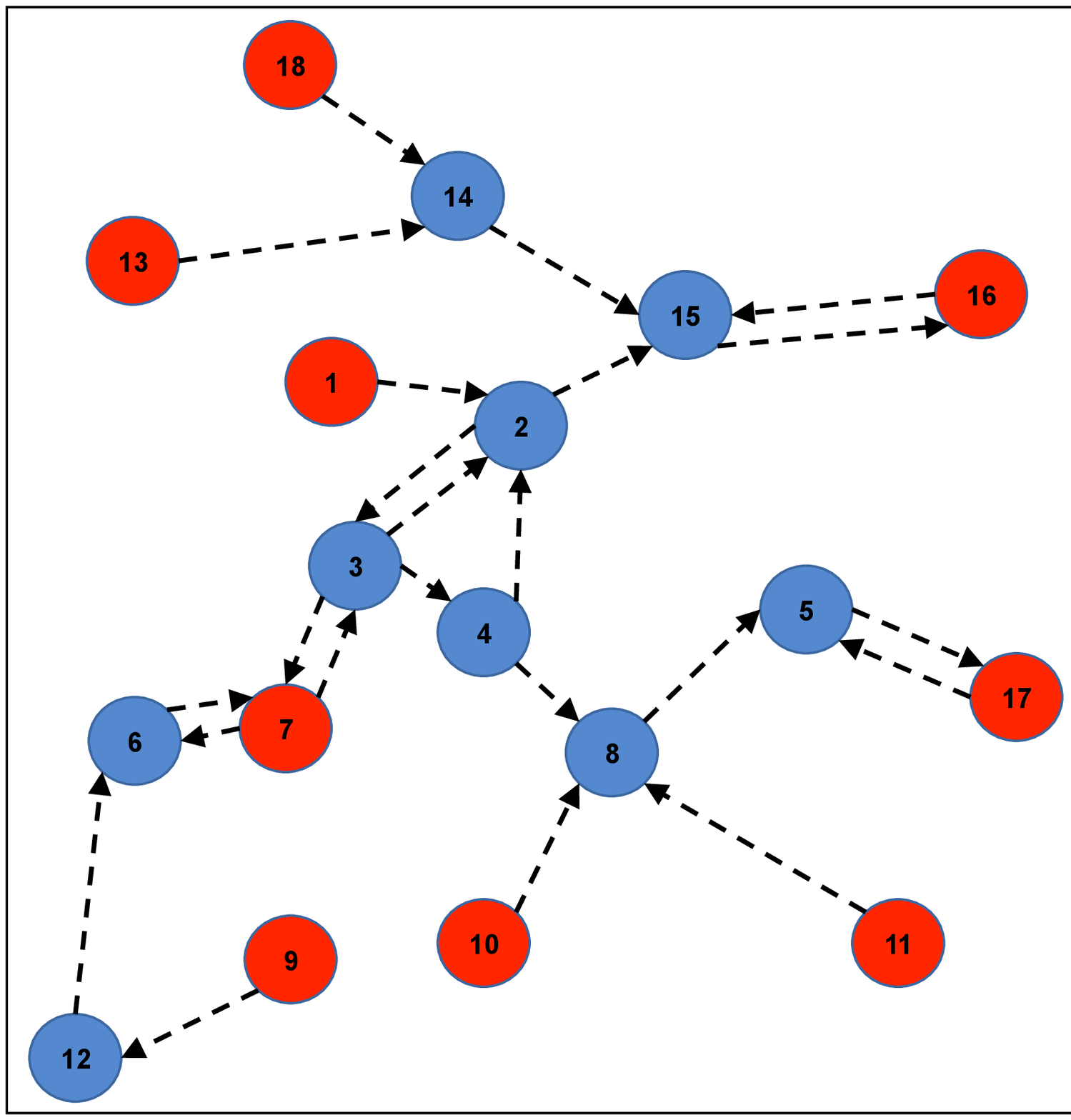}
\caption{Scenario 2: $18-$node network}
\label{fig:sc2}
\end{minipage}\hfill
\begin{minipage}{.32\textwidth}
\centering
\includegraphics[width=1.0\columnwidth]{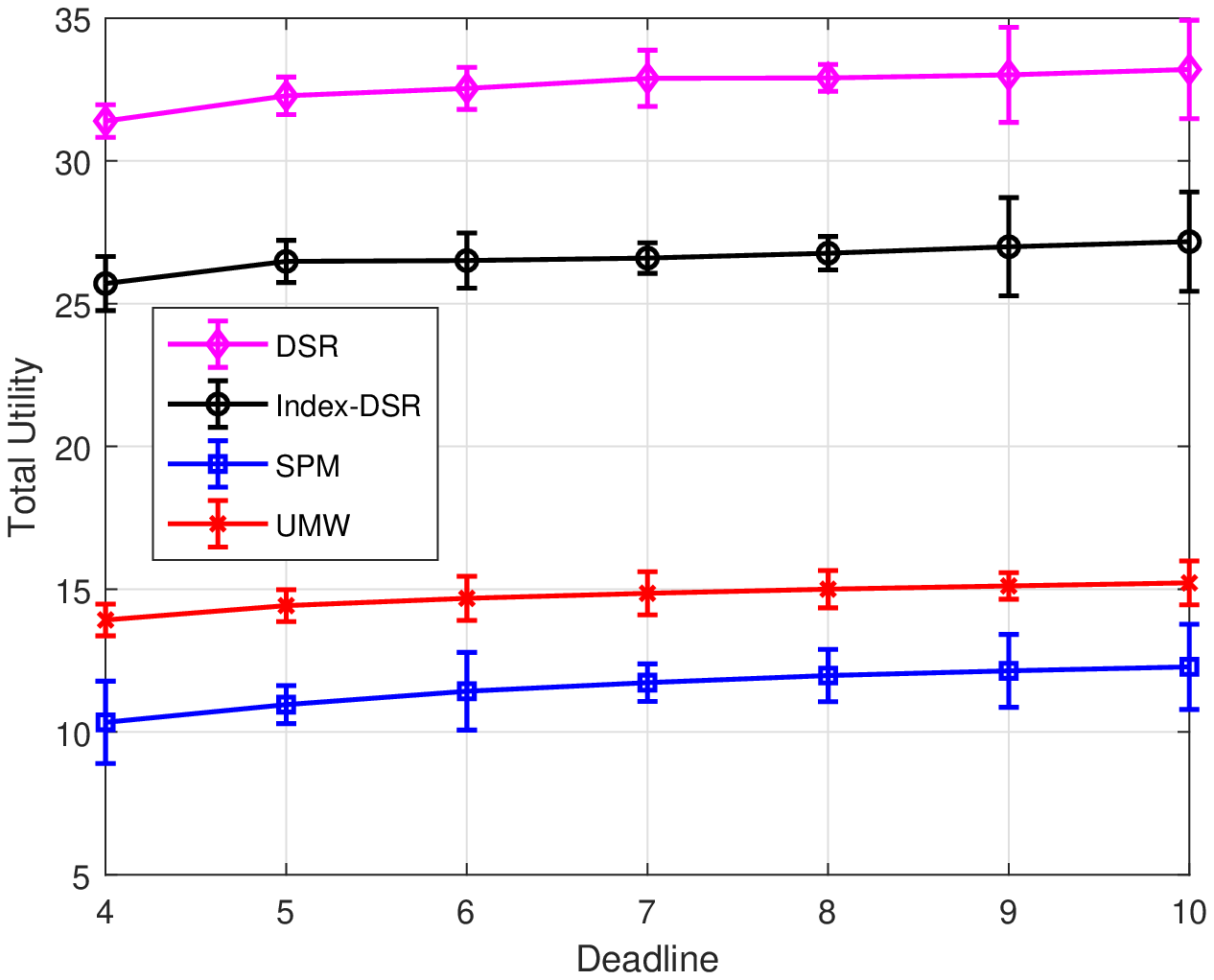}
\caption{Scenario 2: Linear Utility Function}% for scenario in figure (\ref{fig:sc2})}
\label{fig:sc2_lin}
\end{minipage}\hfill
\begin{minipage}{.32\textwidth}
\centering
\includegraphics[width=1.0\columnwidth]{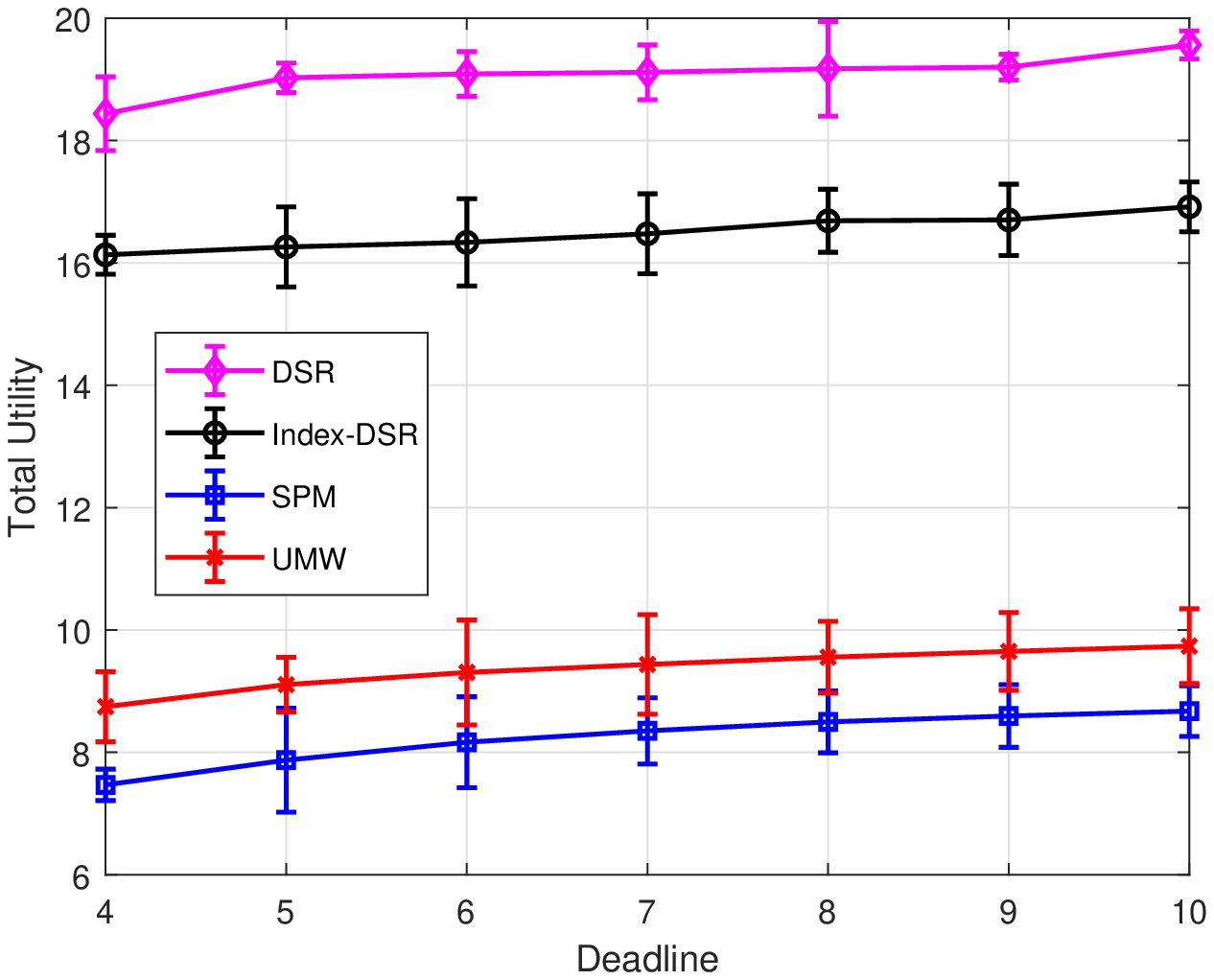}
\caption{Scenario 2: Logarithmic Utility Function}% for scenario in figure (\ref{fig:sc2})}% for odd and even vehicles.}
\label{fig:sc2_log}
\end{minipage}
\end{figure*}

In this section, we present simulation results that compare the performance of our policy against a policy proposed in \cite{umw} called Universal Max-Weight (UMW), and a policy proposed by Sinha, Paschos, and Modiano in \cite{Sinha16-2} that we call SPM.   We first provide a brief description of these two policies, and then present our simulation settings and results.

\subsection{Overview of UMW and SPM}

The UMW policy solves the problem of throughput-optimal packet dissemination in a network with arbitrary topology with different types of traffic, e.g., unicast, multicast and broadcast.  %UMW is considered to be an online dynamic policy consisting routing and packet scheduling.  
In both the centralized and distributed versions of UMW, the route of each packet is decided at the origin.  This route is a weighted tree that is constructed using the edge weights at time of decision at the origin.  Hence, if the route of the packet turns to be inappropriate during packet dissemination, it cannot be modified.   Although this policy also has a heuristic version that can be implemented in a distributed fashion, we consider our comparison against the centralized version, which has better performance than the distributed one.  %Besides, we only consider broadcast traffic in the network.   Furthermore, route decision at origin requires significant information exchange.

The SPM policy is designed specifically for throughput optimal broadcast.  SPM is a virtual-queue based algorithm, where virtual-queues are defined for subsets of nodes.   These virtual queues keep track of a kind of backpressure, while accounting for the fact that packets are duplicated in the broadcast regime.  A feature of this work is that each slot is sub-divided into $L$ minislots, where $L$ is the number of links in the network, and a random link is activated in each mini-slot.  Here, a packet may be retransmitted multiple times over the mini-slots comprising a slot (i.e., it could potentially reach all nodes in just one slot).  To ensure consistency with the slot model, we modify this algorithm to only allow packet state updates at each slot, rather than at each mini-slot.

% under a \emph{mini time slot} model.  Here, each time slot consists of $L$ mini time slots, where $L$ is the number of links.  In each mini time slot, the system randomly picks a link to transmit, and the link employs scheduling policies to determine which packets to transmit.  It is clear that the mini-slot model is not directly applicable to the actual network in which all links may transmit only once per time slot. However, they update buffer status at the end of each mini slot in their original implementation. This results in packet dissemination through the whole network in one slot, which is not applicable. Hence, we keep the mini slot transmission format, but update the status of buffers at the end of each time slot in our simulations. 

\subsection{Simulation Settings and Results}

In this study, we consider two different simulation scenarios motivated by designs for IAB network deployments~\cite{iab,iab2}.  Here, we have two kinds of nodes, namely, (i) gateway nodes with fiber drops (shown in red), and (ii)  wireless-only nodes with mm-wave backhaul (shown as blue nodes).   We assume that gateways communicate reliably between each other with zero latency, since they are connected to the same backend switch (consistent with IAB architecture).   The two scenarios represent different levels of gateway availability.    The first scenario is a an $11-$node network with $2$ fiber drops as in figure (\ref{fig:sc1}), while the second scenario is an $18-$ node IAB network with $9$ fiber drops as shown in figure (\ref{fig:sc2}).  Hence, Scenario~1 is illustrative performance in a network with multiple wireless hops, whereas Scenario~2 illustrates performance in a more densely connected network.

In both scenarios, there are two broadcast flows. One of the flows originates at a  fiber-connected gateway node, and the other one from a wireless-only node.  For each link $l$, $P_l$ is randomly chosen from $[0.5, 1.0],$ and $T_l$ is randomly chosen from $[1, 5]$.  Each flow generates packets according to a Poisson random process, where source node of flow $1$ has a mean arrival rate of $1.5$ packets per time slot, and source node of second flow has a mean arrival rate of $2$ packets per time slot. Since UWM and SPM only aim to maximize throughput, we first consider a linear utility function $U_{n,f}(\mu_{n,f})=\mu_{n,f}$ to make a fair comparison. In this case, the total utility of the system is the same as the total timely-throughputs.  In a second case, we also consider a logarithmic utility function $U_{n,f}(\mu_{n,f})=\log (\mu_{n,f}+1),$ which models the idea that the utility of the end user might be a non-negative, concave and increasing function of timely throughput.    We assume that the two flows have the same deadline of $D$ time slots, and vary $D$ from 4 to 10.  We test four the optimal DSR protocol (for the relaxed problem), the Index-DSR protocol, the UWM policy, and the SPM policy.

The simulation results for the linear utility function and the logarithmic utility function for scenario in figure (\ref{fig:sc1}) are shown in figures (\ref{fig:sc1_lin}) and (\ref{fig:sc1_log}), respectively.  The performance of DSR is an upper bound, since it is the optimal solution under a relaxed constraint.   The Index-DSR protocol outperforms UWM, possibly because of more dynamic routing of each packet under Index-DSR.  This also shows that UMW might be providing bursty service to nodes, since deadlines are often violated and packets are dropped, leading to poor throughput.   The Index-DSR policy outperforms SPM in all cases despite the assumption taht SPM can compute the reachable subgraph for each packet instantly.
%ability of SPM to know the reachable set of each packet instantly.  

The results for second IAB scenario, depicted in (\ref{fig:sc2}) shows similar results in terms performance of DSR-based algorithms for much the same reasons specified above.   
%The reasons of outperforming DSR and Index-DSR are two fold. The first reason is adaptability of DSR based algorithms against routing issues. Routing issues arise when the links are not reliable and the routes need to be corrected during packet dissemination. DSR based algorithms decide routing at each node. While, UMW decides the routes at the origin and sticks to it. On the other hand, SPM does not consider reliability of links in their algorithm.
%
%The second reason is triggering SPM policy. This policy is utilizing mini slot time format. In this case, a packet is transmitted in a mini slot. However, if the status of buffers are updated after each mini slot. Then, a packet is able to go through the whole network in a slot, which is the original implementation in their paper. But, in this comparison, we update the status of buffers at the end of each slot. Hence, that would be another reason why DSR based algorithms outperformed SPM and UMW in those networks.
%
However, results of figures (\ref{fig:sc2_lin}) and (\ref{fig:sc2_log}) shows that UMW has better performance than SPM, unlike the results obtained in (\ref{fig:sc1_lin}) and (\ref{fig:sc2_log}). This result appears to be due to the density of the network.  The UMW policy manages to deliver more unexpired packets to the destination since it has to traverse fewer hops.  SPM is also handicapped by the fact that we force it to obey a slot-by-slot state update model like all the other protocols (although we do allow it to utilize its minislot-based transmission model).   Ultimately, these results demonstrate the efficiency and flexibility of the DSR protocol.

\section{Conclusion}
\label{section:conclusion} 
In this paper, we studied the problem of broadcasting real-time flows with hard per-packet deadlines in a multi-hop wireless network.  
This problem is computationally complex because of the need to solve an MDP over the network graph.
We relax the problem by using average link utilization constraints, and come up with a novel decomposition approach that enables its solution in a distributed fashion.    
We propose the {DSR} algorithm that maximizes the total timely-utility, ensuring that delay guarantees are met.  
The algorithm has a low complexity, and has a really low overhead for coordination among nodes. 
We also develop a simple \emph{index} policy based on DSR that is able to meet hard link utilization constraints.     
%Furthermore, it is shown it achieves higher amount of total utility compared to SPM both theoretically and by simulation results.  
We simulate the variants of the algorithm with different utility functions, comparing against several recent algorithms that are throughput optimal.  
In all cases, {DSR} and the index policy have a better performance in terms of total timely-utility.  
The results indicate that design throughput and delay optimality are fundamentally different, but that it is possible to come up with low complexity and near-optimal solutions in the delay-constrained case.

%Then, \textit{DSR} is compared with SPM with two different utility functions.  
%In both cases, \textit{DSR} has a better performance in terms of total utility.  However, \textbf{Relaxed Utility} policy achieves higher amount of total utility in network.

\bibliographystyle{ieeetr}
\bibliography{reference}

\clearpage

\section{Appendix}
\label{sec:apen}

\subsection{Proof of Lemma \ref{lem:gamma}}

Consider two epoch-wise stationary policies $\omega_1^+$, $\omega_2^+$, and a number $0<a<1$, we will show that there exists an epoch-wise stationary policy $\omega_a^+$ such that $\gamma^{\omega_a^+}=a\gamma^{\omega_1^+}+(1-a)\gamma^{\omega_2^+}$.

We construct $\omega_a^+$ as follows: In epoch $i$, if $\lfloor ai\rfloor>\lfloor a(i-1)\rfloor$, then $\omega_a^+$ uses $\omega_1[\lfloor ai\rfloor]$ in epoch $i$. Otherwise, $\omega_a^+$ uses $\omega_2[i-\lfloor ai\rfloor]$ in epoch $i$.

It is straightforward to check that, in the first $I$ epochs, $\omega_a^+$ consists of the first $\lfloor aI \rfloor$ stationary policies from $\omega_1^+$ and the first $\lceil (1-a)I\rceil$ stationary policies from $\omega_2^+$. Therefore, $\gamma^{\omega_a^+}=a\gamma^{\omega_1^+}+(1-a)\gamma^{\omega_2^+}$.

\subsection{Proof of Theorem \ref{th:rel_utility}}

Since \textbf{Relaxed Utility} is convex, we only need to check the Slater's condition:
\begin{enumerate}
\item $\Gamma \neq \emptyset$.
\item Constraint (\ref{eq:relaxed3}) is a linear inequality. 
\item Consider the policy that never transmits any packets. Under this policy, the number of transmissions over link $l$ is 0, which is strictly less than $T_l$, for all $l$.
\end{enumerate}
Therefore, Slater's condition holds, and the proof is complete.

\subsection{Proof of Theorem \ref{th:lambda}}

For any arbitrary $\lambda'$:
\begin{align*}
&\mathcal{D}(\lambda') = \max_{\gamma} \mathcal{L}(\gamma,\lambda') \geq \mathcal{L}(\gamma(\lambda),\lambda')\\ =&\sum_{n=1}^N\sum_{f=1}^FU_{n,f}(\mu_{n,f}(\lambda))-\sum_{l=1}^L\lambda'_l(\sum_{f=1}^F\bar{\epsilon}_{l,f}(\lambda)- T_l)\\
=&\sum_{n=1}^N\sum_{f=1}^FU_{n,f}(\mu_{n,f}(\lambda))-\sum_{l=1}^L\lambda_l(\sum_{f=1}^F\bar{\epsilon}_{l,f}(\lambda)- T_l)\\
&+\sum_{l=1}^L(\lambda_l-\lambda'_l)(\sum_{f=1}^F\bar{\epsilon}_{l,f}(\lambda)- T_l)\\
=& \mathcal{D}(\lambda)+(\lambda'-\lambda)\cdot [T_l-\sum_{f=1}^F\bar{\epsilon}_{l,f}(\lambda)].
\end{align*}
Thus, $[T_l-\sum_{f=1}^F\bar{\epsilon}_{l,f}(\lambda)]$ is a subgradient of $\mathcal{D}(\lambda)$.

\end{document}